\documentclass[final]{dmtcs-episciences}

\usepackage[utf8]{inputenc} 
\usepackage{xcolor} 

\usepackage{color} 
\usepackage{graphicx, import} 
\usepackage{tikz}
\usepackage{subcaption}

\usepackage{amsmath} 
\usepackage{amssymb} 
\usepackage{amsthm} 
\usepackage{dsfont} 
\usepackage{thm-restate} 

\newcommand\bk[1]{\boldsymbol R_{%
  \ifnum #1<2 {%
  \ifnum #1<1 {%
    \,\leq %
  }\else{%
    \,\neg%
  }\fi%
  }\else{%
    \,\cup%
  }\fi
}}
\newcommand\rrule[1]{$\bk {#1}$}

\newcommand\pk[1]{\boldsymbol P_{%
  \ifnum #1<2 {%
  \ifnum #1<1 {%
    \geq %
  }\else{%
    \cap%
  }\fi%
  }\else{%
    \setminus%
  }\fi
}}
\newcommand\prule[1]{$\pk {#1}$}

\newcommand\st{\;|\;}

\newcommand\relbot[1][r]{\mathrel{ \bot_{#1} }}

\newcommand\veerel{\mathrel{\vee}}
\newcommand\wedgerel{\mathrel{\wedge}}

\newcommand\EE{\mathcal E}

\newcommand\OO{\mathcal O}
\newcommand\GFF{\mathbb F_2}
\newcommand\ZZ{\mathds Z}
\newcommand\KK{\mathds K}
\newcommand\KKo{\mathds K_r^\circ}

\newcommand\clh[1]{{<\!\!#1\!\!>_r}}
\newcommand\clho[1]{{<\!\!#1\!\!>^\circ_r}}

\theoremstyle{plain} 
\newtheorem{theorem}{Theorem}

\newtheorem{lemma}{Lemma}
\newtheorem{corollary}{Corollary}
\newtheorem{observation}{Observation}

\theoremstyle{definition} 
\newtheorem{definition}{Definition}

\theoremstyle{remark} 

\newtheorem{example}{Example}

\title[Hypergraphs with Polynomial Representation: Introducing $r$-splits]{Hypergraphs with Polynomial Representation: Introducing $r$-splits \thanks{This work was supported by the Agence Nationale de la Recherche under grant ANR-20-CE23-0002}}
\author[Fran\c cois Pitois et al.]{%
    Fran\c cois Pitois\affiliationmark{1,2}
    \and Mohammed Haddad\affiliationmark{1}
    \and Hamida Seba\affiliationmark{1}
    \and Olivier Togni\affiliationmark{2}}
\affiliation{
    LIRIS, Université Lyon 1, France\\
    LIB, Université de Bourgogne, France}

\begin{document}
\publicationdata{vol. 25:3 special issue ICGT'22}{2023}{4}{10.46298/dmtcs.10751}{2022-12-29; 2022-12-29; 2023-06-01}{2023-11-02}

\maketitle

\begin{abstract}
Inspired by the split decomposition of graphs and rank-width, we introduce the notion of $r$-splits. We focus on the family of $r$-splits of a graph of order $n$, and we prove that it forms a hypergraph with several properties.
We prove that such hypergraphs can be represented using only $\OO(n^{r+1})$ of its hyperedges, despite its potentially exponential number of hyperedges. We also prove that there exist hypergraphs that need at least $\Omega(n^r)$ hyperedges to be represented, using a generalization of set orthogonality.
\end{abstract}

\section{Introduction}
\label{section:intro}

Graph decomposition is a major aspect of graph theory, it is mainly used to run efficient algorithms and solve combinatorial problems on  graphs that can be well decomposed.
A decomposition of a graph is an alternative way to represent a graph, usually to highlight some structure in the graph. One large family of decompositions is that of width decomposition, which aims at decomposing a graph while minimizing a parameter, and such that this decomposition can be used to run efficient algorithms under this parameterization \cite{hlinveny2008width}.
For example, tree-width decomposes a graph into a tree structure that minimizes the size of some bags \cite{robertson1986graph}; rank-width decomposes a graph into an unrooted binary tree structure such that each edge represents a cut of bounded rank over $\GFF$ \cite{oum2005rank, oum2006approximating}; clique-width decomposes a graph using some allowed operations between a bounded number of classes of vertices \cite{courcelle1993handle}; and more recently, twin-width decomposes a graph by mimicking cograph decomposition and minimizing the number of errors through the decomposition \cite{bonnet2021twin}.
Some other decompositions include modular decomposition \cite{habib2010survey} and split decomposition \cite{cunningham1980combinatorial, cunningham1982decomposition}.
Here, we focus on \emph{split decomposition}. Informally, a split is a cut (or 2-partition) of a graph that looks like a complete bipartite graph plus one stable set on each side of the cut.

\begin{figure}[ht]
\centering
\begin{tikzpicture}[inner sep=0pt,thick,
dot/.style={fill=black,circle,minimum size=4pt}]
    \node[dot] (a) at (1,1) {};
    \node[dot] (b) at (1,1.5) {};
    \node[dot] (c) at (1,2) {};
    \node[dot] (d) at (1,2.5) {};
    \node[dot] (e) at (1,3) {};

    \node[dot] (a2) at (2,1) {};
    \node[dot] (b2) at (2,1.5) {};
    \node[dot] (c2) at (2,2) {};
    \node[dot] (d2) at (2,2.5) {};
    \node[dot] (e2) at (2,3) {};

    \node (p1) at (1.5,0.5) {};
    \node (p2) at (1.5,3.5) {};

    \path (d) edge (b2);
    \path (d) edge (c2);
    \path (d) edge (d2);
    \path (d) edge (e2);

    \path (e) edge (b2);
    \path (e) edge (c2);
    \path (e) edge (d2);
    \path (e) edge (e2);

    \draw (a) to[bend left=50] (c);
    \draw (b) to[bend left=50] (d);
    \draw (d) to[bend left=50] (e);

    \draw (a2) to[bend right=40] (c2);
    \draw (a2) to[bend right=50] (d2);
    \draw (b2) to[bend right=50] (e2);
    \draw (c2) to[bend right=30] (d2);

    \draw[dashed] (p1) to (p2);

\end{tikzpicture}
\caption{An example of a split}
\label{fig:split}
\end{figure}

It was improved in \cite{gioan2012split} to give birth to a decomposition that both represents all the edges of a graph together with all its splits.
The split decomposition of a graph is a powerful tool. For example, it is used to recognize distance-hereditary graphs \cite{gioan2007dynamic}.
However, there is a family of graphs, called \emph{prime graphs}, that have no splits besides trivial ones. A trivial split is a split with one part of size 0 or 1. It is called trivial because a 2-partition with one part of size 0 or 1 is always a split by definition. This means that these techniques cannot be used on prime graphs.
The goal of this paper is to extend this kind of decompositions to prime graphs by generalizing the definition of a split as well as related concepts such as symmetric crossing families \cite{cunningham1980combinatorial} and orthogonality~\cite{charbit2012linear}.
We take inspiration from the rank-width decomposition to generalize splits into $r$-splits.
The family of all splits of a graph has the property of being a \emph{symmetric crossing family}. Hence, our approach consists in proving some properties about the family of $r$-splits of a graph.
We also generalize \emph{orthogonality}. Usually, orthogonality is defined on partitive families and modular decomposition~\cite{habib2010survey, charbit2012linear}. Since these notions are very close to symmetric crossing families and split decomposition~\cite{bui2012tree}, we allow ourselves to extend this notion and the associated vocabulary.
In~\cite{charbit2012linear}, two sets of vertices are orthogonal (or non-overlapping) if one is included in the other or if their intersection is empty. Orthogonality can be used to compute the split decomposition in linear time~\cite{charbit2012linear}.
In~\cite{edmonds1977min}, two sets of vertices are non-crossing if they are non-overlapping or if their union is the set of all vertices of the graph. This notion was used to define and prove properties of split decompositions in the first place~\cite{cunningham1980combinatorial}.
In this paper, we generalize the notion of being non-crossing, and give it the name $r$-orthogonality. We use it to prove lower bounds, and we hope it could be used to define a general $r$-split decomposition.

\subsection{General definitions}
\label{section:general-def}


A split can be defined as a cut of rank at most 1. In a natural way, we introduce an $r$-split as a cut of rank at most $r$.

\begin{definition}
 \label{def:cutrank}
  Let $G=(V,E)$ be a graph. Let $(X, \overline X)$ be a cut (\textit{i.e.}, a 2-partition of $V$). The \emph{rank} of this cut, noted $\rho(X)$, is equal to the rank of the adjacency matrix of $G$ where the rows are restricted to $X$, and the columns are restricted to $\overline X$.
  This matrix is called the $[X, \overline X]$-adjacency matrix of $G$, and is noted $A(G)[X, \overline X]$.
  The rank is computed over the finite field of two elements $\GFF$.
  The cut $(X, \overline X)$ is an \emph{$r$-split} if $\rho (X) \leq r$.
  For convenience, we identify the cut $(X, \overline X)$ with the set of vertices $X$.
\end{definition}

Since the rank of a matrix is always less than or equal to the number of its rows and  less than or equal to the number of its columns, we know that $\rho(X) \leq \min(|X|, |\overline X|)$.
This motivates us to focus on the case where this inequality becomes an equality.
\begin{definition}
  \label{def:trivial}
  A cut $(X, \overline X)$  is said to be \emph{trivial} if $\rho(X) = \min(|X|, |\overline X|) $.
\end{definition}
In other words, the cut $(X, \overline X)$ is trivial if the $[X, \overline X]$-adjacency matrix of $G$ has full-rank in $\GFF$.

\begin{figure}[ht]
\centering
\begin{subfigure}{0.6\textwidth}
  \centering
    \begin{tikzpicture}
    \node[shape=circle,draw=black] (A) at (2,4) {a};
    \node[shape=circle,draw=black] (B) at (2,2) {b};
    \node[shape=circle,draw=black] (C) at (2,0) {c};
    \node[shape=circle,draw=black] (D) at (0,3) {d};
    \node[shape=circle,draw=black] (E) at (0,1) {e};
    \node[shape=circle,draw=black] (F) at (5,4) {f};
    \node[shape=circle,draw=black] (G) at (5,2) {g};
    \node[shape=circle,draw=black] (H) at (5,0) {h};
    \node[shape=circle,draw=black] (I) at (7,2) {i};

    \path (A) edge (F);
    \path (A) edge (G);
    \path (B) edge (F);
    \path (B) edge (H);
    \path (C) edge (G);
    \path (C) edge (H);

    \path (A) edge (B);
    \path (B) edge (D);
    \path (C) edge (E);
    \path (B) edge (E);
    \path (G) edge (I);
    \end{tikzpicture}
\end{subfigure}%
\begin{subfigure}{0.4\textwidth}
  \centering
  \begin{align*}
    \begin{array}{cc}
        &
        \arraycolsep=5pt
        \begin{array}{ccccccccc}
            a&b&c&d&e&f&g&h&i
        \end{array}
        \\
        \begin{matrix}
            a\\b\\c\\d\\e\\f\\g\\h\\i
        \end{matrix}
        &
        \begin{bmatrix}
            0&1&0&0&0&1&1&0&0\\
            1&0&0&1&1&1&0&1&0\\
            0&0&0&0&1&0&1&1&0\\
            0&1&0&0&0&0&0&0&0\\
            0&1&1&0&0&0&0&0&0\\
            1&1&0&0&0&0&0&0&0\\
            1&0&1&0&0&0&0&0&1\\
            0&1&1&0&0&0&0&0&0\\
            0&0&0&0&0&0&1&0&0
        \end{bmatrix}
    \end{array}
  \end{align*}
\end{subfigure}
\caption{A graph and its adjacency matrix}
\label{fig:graph}
\end{figure}

\begin{figure}[ht]
    \centering
    \begin{align*}
    A(G)[\{a, b, c, d, e\},\{f, g, h, i\})] =
    \begin{array}{cc}
        &
        \arraycolsep=5pt
        \begin{array}{ccccccccc}
            f&g&h&i
        \end{array}
        \\
        \begin{matrix}
            a\\b\\c\\d\\e
        \end{matrix}
        &
        \begin{bmatrix}
            1&1&0&0\\
            1&0&1&0\\
            0&1&1&0\\
            0&0&0&0\\
            0&0&0&0\\
        \end{bmatrix}
    \end{array}
    \end{align*}
    \caption{The cut $(\{a, b, c, d, e\},\{f, g, h, i\})$ is a 2-split}
    \label{fig:r-split}
\end{figure}

\begin{example}
    In the graph of Figure \ref{fig:graph}, the set of vertices $\{a, b, c, d, e\}$ is a 2-split, as the rank of the  matrix given in Figure \ref{fig:r-split} is 2. Note that the rank is computed over $\GFF$. This means that additions and multiplications are done modulo 2. For example, the first three lines of the matrix
    ($[1\ \ 1\ \ 0\ \ 0]$, $[1\ \ 0\ \ 1\ \ 0]$ and $[0\ \ 1\ \ 1\ \ 0]$)
    sum to $[0\ \ 0\ \ 0\ \ 0]$, since $1+1=0 \mod 2$.

\end{example}

We focus on the set of $r$-splits of a graph for a fixed $r$. Function $\rho$ has several properties~\cite{oum2006approximating} which carry over to $r$-splits.
The first one is a very well-known property that comes directly from properties of the rank.

\begin{lemma}
    \label{lem:rho-easy}
    For all $X \subseteq V $, we have $\rho(X) \leq |X|$ and $\rho(X)=\rho(V \setminus X)$.
\end{lemma}


Another useful property of $\rho$ is \emph{submodularity}, which is stated as follows:

\begin{lemma}[\cite{oum2006approximating}]
  \label{lem:rho}
  For all $X, Y \subseteq V$, we have $\rho(X \cup Y) + \rho(X \cap Y) \leq \rho(X) + \rho(Y) $.
\end{lemma}

Furthermore, we need the notion of \emph{$r$-rank connectivity}, which is defined as follows:

\begin{definition}[\cite{oum2020rank}]
  \label{def:connected}
  A graph $G$ is \emph{$r$-rank connected} if each $k$-split for $k<r$ is trivial.
\end{definition}

This definition allows to fully characterize the trivial $r$-splits of a graph.
As a consequence of the definition of trivial $r$-split and $r$-rank connectivity, we get the following lemma:

\begin{lemma}
    Let $G$ be a graph.
    If a set of vertices $X$ is a trivial $r$-split, then $|X| \leq r$ or $|\overline X| \leq r$.
    If $G$ is $r$-rank connected, this implication is an equivalence.
\end{lemma}

An another direct consequence of $r$-rank connectivity is this lemma:

\begin{lemma}
    Let $G$ be an $r$-rank connected graph.
    Let $X$ be a set of vertices such that $r \leq |X| \leq |V|-r$.
    Then $\rho(X) \geq r$.
\end{lemma}

\begin{proof}
    Suppose for a contradiction that $\rho(X) < r$.
    This means that $X$ is a $(r-1)$-split.
    By $r$-rank connectivity, $X$ is trivial.
    By definition of being trivial, $\rho(X) = \min(|X|, |\overline X|)$.
    Hence, $ \min(|X|, |\overline X|) < r $, which contradicts the fact that $r \leq |X| \leq |V|-r$.
\end{proof}

Combined with submodularity, we get the following lemma:

\begin{lemma}
  \label{lem:split-union}
  If $X$ and $Y$ are two $r$-splits of an $r$-rank connected graph $G$ and if $|X \cap Y| \geq r$, then $X \cup Y$ is also an $r$-split.
\end{lemma}

\begin{proof}
  First, if $|X \cup Y| \leq r$ or $|\overline{X \cup Y}| \leq r$,
  then $\rho(X \cup Y) \leq \min(|X \cup Y|, |\overline {X \cup Y}|) \leq r$, meaning that $X \cup Y$ is an $r$-split.

  Then, if $|\overline{X \cap Y}| \leq r$, it means that $|\overline{X \cup Y}| \leq |\overline{X \cap Y}| \leq r$, so we are in the same case as previously, and  $X \cup Y$ is an $r$-split.

  Otherwise, the number of vertices of $X\cup Y$ and $X\cap Y$ is both between $r$ and $|V|-r$.
  Hence, by $r$-rank connectivity of $G$, $\rho(X \cup Y) \geq r$ and $\rho(X \cap Y) \geq r$.
  Since $X$ and $Y$ are $r$-splits, we know that $\rho(X) \leq r$ and $\rho(Y) \leq r$. By submodularity, we deduce that $\rho(X \cup Y) + \rho(X \cap Y) \leq 2r$.
  All in all, $\rho(X \cup Y) = \rho(X \cap Y) = r$, which concludes the proof.
\end{proof}

This property is powerful as it organizes the set of all $r$-splits of a graph.
In fact, we show in Section \ref{section:proof-1} that given an $r$-rank connected graph $G$ with $n$ vertices,
there exists a subset of $r$-splits of size $\mathcal O(n^{r+1})$ that fully characterizes the whole set of $r$-splits.

\subsection{Hypergraph of \texorpdfstring{$r$}{r}-splits}

Let $G=(V,E)$ be an $r$-rank connected graph with vertex set $V=[n]$.
Let us denote by $H_r(G)$ the hypergraph whose set of vertices is the same as $G$, namely $V(H_r(G))=[n]$, and whose set of hyperedges $\EE$ is the set of all $r$-splits of $G$.
From Lemmas \ref{lem:rho-easy} and \ref{lem:rho}, we know that $H_r(G)$ satisfies the following properties:
(1) if $A \in \EE$, then $V\setminus A \in \EE$;
(2) for every set of vertices $X\subseteq V$, if $|X| \leq r$, then $X \in\EE$;
(3) if $A,B \in \EE$ and $|A \cap B| \geq r$, then $A \cup B \in \EE$.
Therefore, we consider hypergraphs that satisfy these properties.

Please note that in this paper, since every hypergraph has $V=[n]$ as set of vertices, we identify a hypergraph $H=(V,\EE)$ with its set of hyperedges $\EE$. This means, for example, that we write $A\in\EE$ or $A\in H$ to denote a hyperedge $A$ of $H$, and we can denote by $\{A, B\}$ the hypergraph with set of vertices $V=[n]$ and set of hyperedges $\{A, B\}$.

\begin{definition}
  \label{def:krn}
  Let $\KK_r(n)$ be the class of hypergraphs with set of vertices $V=[n]$ and set of hyperedges $\EE$ that satisfies:
  \begin{itemize}
    \item[$\bk 0:$] For every set of vertices $X\subseteq V$, if $|X| \leq r$, then $X \in\EE$.
    \item[$\bk 1:$] If $A \in \EE$, then $V\setminus A \in \EE$.
    \item[$\bk 2:$] If $A,B \in \EE$ and $|A \cap B| \geq r$, then $A \cup B \in \EE$.
  \end{itemize}
\end{definition}

By combining the three above rules, we can deduce other similar properties.
\begin{lemma}
  \label{lem:krn}
  Let $H\in\KK_r(n)$ be a hypergraph with set of vertices $V=[n]$ and set of hyperedges $\EE$.
  Then $H$ satisfies:
  \begin{itemize}
    \item[$\pk 0:$] For every set of vertices $X\subseteq V$, if $|X| \geq n-r$, then $X \in\EE$.
    \item[$\pk 1:$] If $A,B \in \EE$ and $|\overline{A \cup B}| \geq r$, then $A \cap B \in \EE$.
    \item[$\pk 2:$] If $A,B \in \EE$ and $|A \setminus B| \geq r$, then $B \setminus A \in \EE$.
  \end{itemize}
\end{lemma}

\begin{proof}
  Property \prule 0 is obtained by combining rules \rrule 0 and \rrule 1, while properties \prule 1 and \prule2 are obtained by combining rule \rrule 2 and rule \rrule 1 multiple times.

  First, let us prove property \prule 1:
  Let $A,B \in \EE$ such that $|\overline{A \cup B}| \geq r$.
  By rule \rrule 1, $\overline A \in \EE$ and  $\overline B \in \EE$.
  We know that $|\overline A \cap \overline B| = |\overline{A \cup B}| \geq r$, so by rule \rrule 2,
  $\overline A \cup \overline B = \overline{A \cap B} \in \EE$.
  Finally, by applying rule \rrule 1 again to $ \overline{A \cap B} \in \EE$, we get that $A \cap B \in \EE$.

  In a similar way, we prove property \prule 2:
  Let $A,B \in \EE$ such that $|A \setminus B| \geq r$.
  By rule \rrule 1, $\overline B \in \EE$.
  We know that $|A \cap \overline B| = |A \setminus B| \geq r$, so by rule \rrule 2,
  $A \cup \overline B = \overline{B \setminus A} \in \EE$.
  Finally, by applying rule \rrule 1 again to $ \overline{B \setminus A} \in \EE$, we get that $B \setminus A \in \EE$.
\end{proof}

We note that for each $r$-rank connected graph $G$ of order $n$, the hypergraph $H_r(G)$ made of all $r$-splits of $G$ belongs to the class $\KK_r(n)$.

The class $\KK_r(n)$ has the property of being a \emph{closure system}. This means that:
(1) the hypergraph with every possible hyperedge belongs to $\KK_r(n)$;
(2) if we take two hypergraphs $H_1, H_2 \in \KK_r(n)$, then the intersection of $H_1$ and $H_2$ is also in $\KK_r(n)$.
We recall that the intersection of two hypergraphs $H_1$ and $H_2$ is the hypergraph whose vertex set is the same as $H_1$ and $H_2$ (namely, $[n]$), and whose hyperedge set is the intersection of the set of hyperedges of $H_1$ and $H_2$.

\begin{lemma}
  \label{lem:closure-system}
  The class $\KK_r(n)$ is a closure system.
\end{lemma}

\begin{proof}
  First, it is trivial that the hypergraph with all possible hyperedges satisfies Definition \ref{def:krn}, meaning that this hypergraph belongs to $\KK_r(n)$.
  Secondly, let $H_1, H_2 \in \KK_r(n)$ and let us prove that $H_1 \cap H_2 \in \KK_r(n)$. To this purpose, let $A, B$ be hyperedges of $H_1 \cap H_2$ and let $X$ be a subset of vertices of $V$, and let us prove that they satisfy rules $\bk 0, \bk 1, \bk 2$ of Definition \ref{def:krn}:
  \begin{itemize}
    \item For $\bk 0$: If $|X| \leq r$, then $X$ is a hyperedge of $H_1$ as $H_1 \in \KK_r(n)$, and $X$ is a hyperedge of $H_2$ as $H_2 \in \KK_r(n)$. Hence, $X$ is a hyperedge of $H_1 \cap H_2$.
    \item For $\bk 1$: As $A$ is a hyperedge of $H_1 \cap H_2$, $A$ is a hyperedge of $H_1$, and $V\setminus A$ is also a hyperedge of $H_1$ as $H_1\in\KK_r(n)$. For the same reason, $V\setminus A$ is a hyperedge of $H_2$. Hence, $V \setminus A$ is a hyperedge of $H_1 \cap H_2$.
    \item For $\bk 2$: If $|A \cap B| \geq r$, then $A \cup B$ is a hyperedge of $H_1$ as $H_1\in\KK_r(n)$. With the same argument, $A \cup B$ is a hyperedge of $H_2$.
    Hence, $A \cup B$ is a hyperedge of $H_1 \cap H_2$.
  \end{itemize}
  In conclusion, $H_1 \cap H_2$ fully satisfies Definition \ref{def:krn}, proving that $H_1 \cap H_2 \in \KK_r(n)$.
\end{proof}

Having a closure system is convenient, as it induces a \emph{closure operator}~\cite{caspard2003lattices}.
In our case, the closure operator is defined as follows:

\begin{definition}
  \label{def:closure}
  Let $H$ be a hypergraph with vertex set $V=[n]$. The \emph{closure} of $H$ in $\KK_r(n)$, denoted $\clh H$, is the hypergraph defined as the intersection of all hypergraphs that contain $H$ and that belong to $\KK_r(n)$.
  A hypergraph $H$ satisfying $\clh H = H$ is called a \emph{closed hypergraph for $\clh \cdot$}, an $r$-closed hypergraph, or simply a closed hypergraph when there is no ambiguity.
\end{definition}

In other words, a set of vertices $A$ is a hyperedge of $\clh H$ if and only if $A$ is a hyperedge of every hypergraph of $\KK_r(n)$ that has $H$ as a sub-hypergraph.
To better understand the closure operator $\clh \cdot$, let us see an example.

\begin{example}
    Let $n=8$ and $r=2$. Let $H$ be the hypergraph with vertex set $V=[n]$ and hyperedge set $\EE = \{ \{ 1,2,3\}, \{ 2,3,4,5\} \}$.
    Then, the hypergraph $\clh H$ is made of the following hyperedges:
    \begin{itemize}
        \item all sets made of 0, 1 or 2 vertices, according to \rrule 0,
        \item all sets made of 6, 7 or 8 vertices, according to \prule 0,
        \item $ \{ 1,2,3\}$ and $ \{ 2,3,4,5\} $, because $H$ must be a sub-hypergraph of $\clh H$,
        \item $ \{ 4,5,6,7,8\}$ and $ \{ 1,6,7,8\} $, according to \rrule 1,
        \item $ \{ 1,2,3,4,5\}$ according to \rrule 2,
        \item $ \{ 6,7,8\}$ according to \rrule 1 applied to $ \{ 1,2,3,4,5\}$.
    \end{itemize}
    This list proves that $\clh H$ must contain all these hyperedges. To prove that $\clh H$ is indeed equal to this set of hyperedges, one must prove that this whole set of hyperedges satisfies rules \rrule 0, \rrule 1 and \rrule 2.
\end{example}

Just like any closure operator, $\clh \cdot $ is \emph{extensive} (for any hypergraph $H$, $H\subseteq \clh H$), \emph{monotone} (for any hypergraphs $H, H'$, if $H \subseteq H'$, then $\clh H \subseteq \clh {H'}$), and \emph{idempotent} (for any $H \in \KK_r(n)$, we have $\clh H = H$)~\cite{caspard2003lattices}.
We can now use $\clh \cdot $ to formalize one of the main theorems of this paper.

\begin{restatable}{thm}{restateOne}
  \label{thm:main-1}
  Given an $r$-rank connected graph $G$ with $n$ vertices, there exists a hypergraph $H$ with $\mathcal O(n^{r+1})$ hyperedges such that $\clh H = H_r(G)$.
\end{restatable}

Section \ref{section:proof-1} is dedicated to the proof of this theorem.

\subsection{Complementary results}
\label{section:complementary}

The closure operator $\clh\cdot$ is the main tool used to represent hypergraphs through this paper. In this subsection, we study it in more detail.
To do so, we first introduce a related closure operator, which in turn defines a relation between the hyperedges of a hypergraph.

\begin{definition}
  \label{def:kro}
  Let $\KKo(n)$ be the class of hypergraphs with set of vertices $V=[n]$ and set of hyperedges $\EE$ that satisfies:
  \begin{itemize}
    \item[$\bk 0:$] For every set of vertices $X\subseteq V$, if $|X| \leq r$, then $X \in\EE$.
    \item[$\bk 1:$] If $A \in \EE$, then $V\setminus A \in \EE$.
  \end{itemize}
\end{definition}

The definition of $\KKo(n)$ is similar to the definition of $\KK_r(n)$. The only difference is that the rule $\bk 2$ is removed.
In a very similar way, we can prove that $\KKo(n)$ is a closure system, and thus we can define the corresponding closure operator $\clho\cdot$.
Now, we define $r$-orthogonal hyperedges as follows:

\begin{definition}
  \label{def:ortho-def}
  Let $A$ and $B$ be two hyperedges of a hypergraph $H$.
  Hyperedges $A$ and $B$ are {\em $r$-orthogonal} if
  $\clh{\{A,B\}} = \clho{\{A,B\}} $.
  Here, recall that $\{A,B\}$ denotes the hypergraph with vertex set $[n]$ and hyperedge set $\{A,B\}$.
\end{definition}

Section \ref{section:orthogonal} is dedicated to the study of $r$-orthogonality.
This relation introduces a class of hypergraphs called $r$-cross-free hypergraphs, defined as follows:

\begin{definition}
  \label{def:cross-free}
  A hypergraph is {\em $r$-cross-free} if each pair of hyperedges of $H$ are $r$-orthogonal.
\end{definition}

Such hypergraphs are interesting as they have a few number of hyperedges, but each of them is important regarding the closure operator $\clh \cdot$.

\begin{restatable}{thm}{restateTwo}
  \label{thm:main-2}
  Let $H$ be a $r$-cross-free hypergraph with $n$ vertices.
  Then the number of hyperedges of $H$ is at most $\OO(n^{r+1})$.
\end{restatable}

\begin{restatable}{thm}{restateThree}
  \label{thm:main-3}
  There exists a hypergraph $H$ with $n$ vertices that is both $r$-cross-free and $r$-closed such that for all sub-hypergraph $H'$ of $H$ satisfying $\clh{H'} = H$, the number of hyperedges of $H'$ is at least $\Omega(n^r)$.
\end{restatable}

Section \ref{section:bound} is dedicated to proving these two theorems.

Theorems \ref{thm:main-1} and \ref{thm:main-3} together prove that some $r$-closed hypergraphs need at least $\Omega(n^r)$ hyperedges to be represented,
while each $r$-closed hypergraph needs at most $\OO(n^{r+1})$ hyperedges to be represented.

\section{Essential hyperedges and polynomial representation}
\label{section:proof-1}

The sketch of the proof of Theorem~\ref{thm:main-2} is as follows: Given a hypergraph $H\in\KK_r$, we define a notion of essential hyperedge of $H$, such that the sub-hypergraph $H'$ made of all essential hyperedges has a closure equal to $H$. By being careful with the definition of essential hyperedges, we ensure that there is no more than $\OO(n^{r+1})$ essential hyperedges. This allows us to conclude.

This is done in three steps:
First, we define what an essential hyperedge is. To do so, we need some lemmas that ensure that this notion is well-defined.
Then, we prove that each hyperedge of the original hypergraph $H$ can be obtained using essential hyperedges. Formally, this means that for each hyperedge $A$ of $H$, there exists a hypergraph $H'$ made of some essential hyperedges such that $A$ is a hyperedge of $\clh {H'}$.
Finally, this allows us to prove that the hypergraph $H'$ made of all essential hyperedges satisfies $H \subseteq \clh{H'}$.
Besides, since $H'$ is a subgraph of $H$, we have an equality.
Hence, we have a result for every hypergraph in $\KK_r$. Since $H_r(G)\in\KK_r$, we can conclude.

\subsection{Definition of essential hyperedges}

We want to define essential hyperedges of a hypergraph $H\in\KK_r$ such that, using them together with rules \rrule 0, \rrule 1, and \rrule 2, we can deduce every other hyperedge of $H$.
Rule \rrule 1 states that if a hyperedge $A$ is in $H$, then its complement $\overline A$ is also in $H$.
This means that only half of the hyperedges are useful and that the other half can be obtained using rule \rrule 1.
For instance, it motivates the fact that only hyperedges $A$ with a number of edges satisfying $|A| \leq n/2$ should be essential.

Our idea is to define a function $\varphi_H$ as follows. We pick a set $X$ of ${r+1}$ vertices.
If there exists a hyperedge $A$ of $H$ that contains $X$ such that $|A| \leq n/2$, then we map $X$ to the smallest such hyperedge (for instance, either $A$ or a subset of $A$).
Otherwise, it means that every hyperedge of $H$ that contains $X$ has more than $n/2$ vertices.
We decide to map $X$ to nothing in this case and to remove $X$ from the domain of $\varphi_H$.

First, we need to prove that, under some conditions, the intersection of two hyperedges $A$ and $B$ of a hypergraph $H\in\KK_r(n)$ is also in $H$.
In $\KK_r(n)$, we have property \prule 1 that states that if $A$ and $B$ are two hyperedges of $H$ such that $|\overline{A \cup B}| \geq r$, then $A \cap B$ is also a hyperedge of $H$.
We show that we can remove the condition $|\overline{A \cup B}| \geq r$ if we consider only hyperedges with less than $n/2$ vertices.

\begin{lemma}
  \label{lem:closed-inter-two}
  Let $H\in\KK_r(n)$. Let $A$ and $B$ be two hyperedges of $H$ such that $|A| \leq n/2$ and $|B| \leq n/2$. Then $A \cap B$ is a hyperedge of $H$.
\end{lemma}

\begin{proof}
  If $|A \cap B| \leq r$, then $A \cap B$ is a hyperedge of $H$ by rule \rrule 0.
  Otherwise, we have $|A \cap B| > r$.
  Hence, $|A \cup B| = |A| + |B| - |A \cap B| \leq n/2 + n/2 - r \leq n-r$.
  This means that $|\overline{A \cup B}| \geq r$.
  By property \prule 1, $A \cap B$ is a hyperedge of $H$.
\end{proof}

We extend this lemma to more than two hyperedges by induction:

\begin{corollary}
  \label{cor:closed-inter}
  Let $H\in\KK_r(n)$. Let $X \subseteq V$ be a set of vertices.
  Let $\mathcal A$ be the set of all hyperedges $A$ of $H$ satisfying $X \subseteq A$ and $|A| \leq n/2$.
  Then, if $\mathcal A$ is nonempty, the intersection of all hyperedges of $\mathcal A$ is also a hyperedge of $\mathcal A$.
\end{corollary}


We are now ready to define the notion of essential hyperedge.

\begin{definition}
  \label{def:function-phi}
  Let $H\in\KK_r(n)$.
  An \emph {essential hyperedge} is a hyperedge of the form $\varphi_H(X)$ for some set of vertices $X\subseteq V$, where the function $\varphi_H$ is defined as follows.
  It is a function that takes as input a set of vertices $X \subseteq V$ of size $|X| = r+1$ and that returns either:
  \begin{itemize}
    \item A hyperedge $A$ of $H$ such that $X \subseteq A$, $|A| \leq n/2$, and such that $A$ is the smallest such hyperedge.
    \item A guaranty that each hyperedge $A$ of $H$ that contains $X$ has size $|A| > n/2$. In this case, $\varphi_H(X)$ is undefined, and $X$ is not in the input space of $\varphi_H$.
  \end{itemize}

This function is well-defined, as the notion of ``smallest such hyperedge" exists thanks to Corollary \ref{cor:closed-inter}.
\end{definition}

\subsection{Each hyperedge is the union of some essential hyperedges}
\label{section:union-essential}

In this section, we prove that if $A$ is a hyperedge of $H$, then there exists a hypergraph $H'$ made of some essential hyperedges such that $A$ is a hyperedge of $\clh {H'}$.
To prove this, we discuss the number of vertices of the hyperedge $A$.
The main case is the case where $r+1 \leq |A| \leq n/2$. Informally, the fact the $|A| \geq r+1$ guarantees that we can pick a set of vertices $X \subseteq A$ of size $|X| = r+1$ to apply $\varphi_H$ to; and the fact that $|A| \leq n/2$ guarantees that $\varphi_H(X)$ is defined, meaning that we have at least one essential hyperedge $\varphi_H(X)$ contained in $A$.
The other cases (namely when $|A| \leq r$ and when $|A| > n/2$) are treated easily afterward.

First, we need a lemma that states that the union of some essential hyperedges of $H$ is a hyperedge (not necessarily essential) of $H$, providing some conditions.
To keep it as general as possible, we do not ask for the hyperedges to be essential, but rather we just ask them to belong to a common set of size at most $n/2$ vertices.

\begin{lemma}
  \label{lem:closed-inter-any}
  Let $H \in \KK_r(n)$
  and let $A_1, \ldots, A_k$ be $k$ hyperedges of $H$.
  If for every $1 \leq i < k$, we have $|A_i \cap A_{i+1}| \geq r$, then $\bigcup_{i=1}^k A_i$ is a hyperedge of $H$.
\end{lemma}

\begin{proof}
  We prove it by induction.
  For $k=1$, the lemma is trivial, and for $k=2$, it corresponds to \rrule 2.
  Suppose the lemma is true for $k$ hyperedges, and let us prove it for $k+1$.
  Let $A' = A_1 \cup \ldots \cup A_k$. $A'$ is a hyperedge of $H$ by induction.
  It remains to prove that $A' \cup A_{k+1}$ is a hyperedge of $H$.
  Since $A_k \subseteq A'$, we have $A_k \cap A_{k+1} \subseteq A' \cap A_{k+1}$, meaning that $|A' \cap A_{k+1}| \geq |A_k \cap A_{k+1}| \geq r$.
  Finally, by rule \rrule 2, $A' \cup A_{k+1}$ is a hyperedge of $H$.
\end{proof}

With this lemma, we can prove the first case of this section, namely that each hyperedge $A$ with a size satisfying $r+1 \leq |A| \leq n/2$ can be written as a union of essential hyperedges.

\begin{corollary}
  \label{cor:exist-list-phi}
  Let $H \in \KK_r(n)$. Let $A$ be a hyperedge of $H$ such that $r+1 \leq |A| \leq n/2$.
  Then there exists a list of sets of vertices $X_1, \ldots, X_k$ and a hypergraph $H'$ such that the hyperedges of $H'$ are exactly $\varphi_H(X_1), \ldots, \varphi_H(X_k)$ and such that $A \in \clh {H'}$.
\end{corollary}

\begin{proof}
  Let us write $A$ as $A = \{u_1, \ldots, u_{|A|}\} $.
  For $1 \leq i \leq |A|-r$, let $X_i = \{u_i, \ldots, u_{i+r}\}$.
  For each $X_i$, we have $|X_i| = r+1$. Furthermore, $\varphi_H(X_i)$ is not undefined, as there exists a hyperedge in $H$ that contains $X_i$ and that has at most $n/2$ vertices, namely $A$.

  Let us prove that $A \in \clh {H'}$ by applying Lemma \ref{lem:closed-inter-any} with hypergraph $\clh {H'}$
  and the $k$ hyperedges $\varphi_H(X_1), \ldots, \varphi_H(X_k)$.
  To be allowed to apply this lemma, we have to check that for all $i$,
  $|\varphi_H(X_i) \cap \varphi_H(X_{i+1})| \geq r$.
  %
  This is true as $X_i \subseteq \varphi_H(X_i)$ and as
  $X_i \cap X_{i+1} = \{ u_{i+1}, \ldots u_{i+r} \}$, which is a set of $r$ vertices.
  Hence, $|\varphi_H(X_i) \cap \varphi_H(X_{i+1})| \geq |X_i \cap X_{i+1}| = r$.

  Therefore, $A \in \clh {H'}$.
\end{proof}

As a consequence, using the rule \rrule 1, we can deal with the case $n/2 \leq |A| \leq n-r-1$:

\begin{corollary}
  \label{cor:exist-list-phi-2}
  Let $H \in \KK_r(n)$. Let $A$ be a hyperedge of $H$ such that $n/2 \leq |A| \leq n-r-1$.
  Then there exists a list of sets of vertices $X_1, \ldots, X_k$ and a hypergraph $H'$ such that the hyperedges of $H'$ are exactly $\varphi_H(X_1), \ldots, \varphi_H(X_k)$ and such that $A \in \clh {H'}$.
\end{corollary}

\begin{proof}
  Let $B = \overline A$. Then $r+1 \leq |B| \leq n/2$. By Corollary \ref{cor:exist-list-phi}, there exists a list of sets of vertices $X_1, \ldots, X_k$ and a hypergraph $H'$ such that the hyperedges of $H'$ are exactly $\varphi_H(X_1), \ldots, \varphi_H(X_k)$ and such that $B \in \clh {H'}$. By rule \rrule 1, $\overline B  \in \clh {H'}$.
\end{proof}

Finally, it remains the cases where  $|A| \leq r$ or $|A| \geq n-r$:

\begin{lemma}
  \label{lem:emptyset-closure}
  Let $H \in \KK_r(n)$. Let $A$ be a hyperedge of $H$ such that $|A| \leq r$ or $|A| \geq n-r$.
  Then $A \in \clh {\emptyset}$, where $\emptyset$ represents the hypergraph with no hyperedge.
\end{lemma}

\begin{proof}
  If $|A| \leq r$, by applying rule \rrule 0 to the empty hypergraph $\emptyset$, we have that $A \in \clh {\emptyset}$.
  If $|A| \geq n-r$, we apply property \prule 0 to the empty hypergraph $\emptyset$ to obtain that $A \in \clh {\emptyset}$.
\end{proof}

All in all, given a hypergraph $H \in \KK_r(n)$ and a hyperedge $A$ of $H$, there exists a list (that may be empty) of sets of vertices $X_1, \ldots, X_k$ and a hypergraph $H'$ such that hyperedges of $H'$ are exactly $\varphi_H(X_1), \ldots, \varphi_H(X_k)$ and such that $A \in \clh {H'}$.

\subsection{Polynomial representation}

To finally prove Theorem \ref{thm:main-1}, it remains to apply the results of Section \ref{section:union-essential} to each hyperedge of the hypergraph $H$.

\begin{lemma}
  \label{lem:poly-closure}
  Given a hypergraph $H \in \KK_r(n)$, there exists a list of sets of vertices $X_1, \ldots, X_k$ and a hypergraph $H'$ such that the hyperedges of $H'$ are exactly $\varphi_H(X_1), \ldots, \varphi_H(X_k)$ and such that $H = \clh {H'}$.
\end{lemma}

\begin{proof}
  For each hyperedge $A_i$ of $H$, we know that there exists a hypergraph $H_i$ such that hyperedges of $H_i$ are of the form $\varphi_H(X_{i_1}), \ldots, \varphi_H(X_{i_k})$ and such that $A_i \in \clh {H_i}$.
  Let $H'$ be the union of all $H_i$, and let us prove that $H = \clh {H'}$.

  First, we have:
  \begin{align*}
    H = \bigcup_i \{A_i\} \subseteq \bigcup_i \clh{H_i} \mathop{\subseteq}\limits_{(1)} \clh{H'}.
  \end{align*}
  Inclusion $(1)$ is due to the monotone property of the closure operator $\clh \cdot$ applied to $H_i$ and $H'$ for all $i$.
  This proves that for every $i$, $\clh{H_i} \subseteq \clh{H'}$.
  Hence, the union of $\clh{H_i}$ is included in $\clh{H'}$.

  Furthermore,
  \begin{align*}
    H' = \bigcup_i H_i  \mathop{\subseteq}\limits_{(2)} \bigcup_i H = H.
  \end{align*}
  Inclusion $(2)$ is due to the extensive property of the closure operator $\clh \cdot$.

  Hence, by the monotone property of the closure operator $\clh \cdot$, $\clh{H'} \subseteq \clh H$,
  and $\clh{H'} \subseteq H $ since $H \in \KK_r(n)$.

  All in all, by double inclusion, $H = \clh {H'}$.
\end{proof}

\begin{corollary}
  \label{cor:poly-closure}
  Given a hypergraph $H \in \KK_r(n)$, there exists a hypergraph $H'$ such that the number of hyperedges of $H'$ is $\OO(n^{r+1})$ and such that $H = \clh {H'}$.
\end{corollary}

\begin{proof}
  By applying Lemma \ref{lem:poly-closure}, there exists $H'$ made of hyperedges of the form $\varphi_H(X_i)$. The number of such hyperedges is $\OO(n^{r+1})$ as every $X_i$ is made of $r+1$ vertices.
\end{proof}

Now, we can prove Theorem \ref{thm:main-1}. Let us recall it.

\begin{theorem}
  Given an $r$-rank connected graph $G$ with $n$ vertices, there exists a hypergraph $H$ with $\mathcal O(n^{r+1})$ hyperedges such that $\clh H = H_r(G)$.
\end{theorem}

\begin{proof}
  Let $G$ be an $r$-rank connected graph $G$ with $n$ vertices. Then, $H_r(G) \in \KK_r(n)$.
  By Corollary \ref{cor:poly-closure}, there exists a hypergraph $H'$ such that the number of hyperedges of $H'$ is $\OO(n^{r+1})$ and such that $H_r(G) = \clh {H'}$.
\end{proof}

\section{Orthogonal hyperedges}
\label{section:orthogonal}

In this section, we generalize the notion of $r$-orthogonal hyperedges.
This notion was introduced in the case $r=1$ under the name ``non-crossing hyperedges" \cite{cunningham1980combinatorial}.
We prefer to use the term ``orthogonal hyperedges", as it is more common, notably in modular decomposition \cite{habib2010survey, charbit2012linear}.
Informally, two hyperedges are orthogonal when they do not contribute to the existence of a lot of new hyperedges in a hypergraph.
The idea is that, given two hyperedges of a hypergraph in $\KK_r(n)$, because of rules \rrule 0, \rrule 1 and \rrule 2, they can imply the existence of a lot of other hyperedges.
Namely, rules \rrule 1 and \rrule 2 can imply up to eight new hyperedges ($A \cup B, A \setminus B, B \setminus A, A \cap B$, and all their complement).
In turn, these new hyperedges can imply the existence of a lot of new hyperedges, which implies at the end an explosion of the number of hyperedges in the whole hypergraph.
Without \rrule 2, the number of new hyperedges is reduced by far, as one hyperedge implies only the existence of another hyperedge (its complement).

Thus, the idea is to control how rule \rrule 2 applies to hyperedges $A$ and $B$, so that the produced hyperedge $A \cup B$ could also be obtained by using other rules, \textit{i.e.}, \rrule 0 and/or \rrule 1.
This notion depends only on the hyperedges $A$ and $B$: either $A \cup B$ can be obtained using only \rrule 0 and/or \rrule 1; or \rrule 2 is needed as well to obtain $A \cup B$.
In the first case, hyperedges $A$ and $B$ are said to be orthogonal.
In the latter case, they are said to cross.
To formalize this, we introduce a new class of hypergraphs and a new closure operator that uses only rules \rrule 0 and \rrule 1, as announced in Section \ref{section:complementary}:

\begin{definition}
  \label{def:kon}
  Let $\KKo(n)$ be the class of hypergraphs with set of vertices $V=[n]$ and set of hyperedges $\EE$ that satisfies:
  \begin{itemize}
    \item[$\bk 0:$] For every set of vertices $X\subseteq V$, if $|X| \leq r$, then $X \in\EE$.
    \item[$\bk 1:$] If $A \in \EE$, then $V\setminus A \in \EE$.
  \end{itemize}
\end{definition}

\begin{lemma}
  \label{lem:kon-closure-system}
  The class $\KKo(n)$ is a closure system.
\end{lemma}

\begin{proof}
    The proof is the same as the proof of Lemma \ref{lem:closure-system}, except we do not need to prove the last point.
\end{proof}

This induces the following closure operator.

\begin{definition}
  \label{def:kon-closure}
  Let $H$ be a hypergraph with vertex set $V=[n]$. The $\KKo$-closure of $H$, denoted $\clho H$, is the hypergraph defined as the intersection of all hypergraphs that contain $H$ and that belong to $\KKo(n)$.
  A hypergraph $H$ satisfying $\clho H = H$ is called a \emph{closed hypergraph for $\clho \cdot$}.
\end{definition}


We can now define the notion of $r$-orthogonality:

\begin{definition}
  \label{def:orthogonal}
  Let $V=[n]$ be a set of vertices and let $A, B \subseteq V$ be two sets of vertices.
  Hyperedges $A$ and $B$ are said to be \emph{$r$-orthogonal} if $\clh{\{A,B\}} = \clho{\{A,B\}} $.
  This is noted $A \relbot B$.
  Recall that $\{A,B\}$ denotes the hypergraph with vertex set $[n]$ and hyperedge set $\{A,B\}$.
\end{definition}
In other words, two sets of vertices $A$ and $B$ are $r$-orthogonal if the smallest hypergraph containing them as hyperedge and satisfying rules \rrule 0 and \rrule 1 is the same as the smallest hypergraph containing them as hyperedge and satisfying rules \rrule 0, \rrule 1 and \rrule 2.
This is a way of guaranteeing that \rrule 2 is useless with regard to $A$ and $B$.

\begin{example}
    \label{ex:cross-free}
    Let $V=[n]$ with $n=12$, and let $r=3$.
    Let $A=\{1,2,3\}$, $B=\{2,3,4,5,6\}$, $C=\{1,2,3,4,5,6\}$ and $D=\{4,5,6,7,8,9\}$.
    One can prove that $A \relbot B$ and $C \relbot D$ using Definition \ref{def:kon-closure}. However, this can be a bit tedious. Hence, we develop the following lemmas to prove some simple equivalences of the relation $\relbot$.
\end{example}

First, let us state some lemmas about the closure operator $\clho\cdot$.
A direct application of the definitions of $\clho\cdot$ gives the following lemma:

\begin{lemma}
  \label{lem:emptyset-hypergraph}
  Let $V=[n]$ and let $\emptyset$ be the empty hypergraph with vertex set $V$.
  Then $\clho \emptyset = \clh \emptyset = \{A \subseteq V \st |A| \leq r \text{ or } |\overline A| \leq r \}$.
\end{lemma}


\begin{lemma}
  \label{lem:singleton-simple}
  Let $V=[n]$ and let $A \subseteq V$ and $B \subseteq V$ be two sets.
  Then we have
  $\clho {\{A,B\}} = \clh \emptyset \cup \{A, B, \overline A, \overline B\}$.
\end{lemma}

\begin{proof}
  Let $H=\{A,B\}$ and $H' =  \clh \emptyset \cup \{A, B, \overline A, \overline B\}$.
  First, let us prove that $\clho H \subseteq H'$. Since $\clho H$ is the smallest hypergraph that has $H$ as a sub-hypergraph and that satisfies rules \rrule 0 and \rrule 1, it suffices to prove that $H'$ has $H$ as a sub-hypergraph and that $H'$ satisfies rules \rrule 0 and \rrule 1:
  \begin{itemize}
    \item $H'$ has hyperedges $A$ and $B$, meaning that $H'$ has $H$ as a sub-hypergraph.
    \item $H'$ satisfies rule \rrule 0 since $\clh \emptyset \subseteq H'$ and $\clh \emptyset$ satisfies \rrule 0 by definition.
    \item $H'$ satisfies rule \rrule 1 since $\clh \emptyset$ satisfies \rrule 1 by definition, and $\{A, B, \overline A, \overline B\}$ satisfies \rrule 1 by construction.
  \end{itemize}

  Now, let us prove that $H' \subseteq \clho H$.
  First, $A$ and $B$ are hyperedges of $\clho H$ since $\clho H$ has $H$ as a sub-hypergraph.
  Then, $\overline A$ and $\overline B$ are hyperedges of $\clho H$ since $\clho H$ satisfies rule \rrule 1.
  Finally, $\clh \emptyset = \clho \emptyset \subseteq \clho H$.
\end{proof}

\begin{observation}
    \label{lem:closure-sub}
    Let $H$ be a hypergraph. Then $\clho H \subseteq \clh H$.
\end{observation}


Now, let us understand what it means for two hyperedges to be $r$-orthogonal.

\begin{lemma}
  \label{lem:ortho-equivalence}
  Let $r\geq 0$ and $V=[n]$. Let $A \subseteq V$ and $B \subseteq V$ be two sets.
  We have the following equivalence.
  \begin{align*}
    A \relbot B \iff &(|A \cap B| \geq r \implies A \cup B = A \veerel A \cup B = B  \veerel |\overline A \cap \overline B| \leq r) \wedgerel \\
    &(|\overline A \cap \overline B| \geq r \implies \overline A \cup \overline B = \overline A \veerel \overline A \cup \overline B = \overline B  \veerel |A \cap B| \leq r) \wedgerel \\
    &(|\overline A \cap B| \geq r \implies \overline A \cup B = \overline A \veerel \overline A \cup B = B  \veerel |A \cap \overline B| \leq r) \wedgerel \\
    &(|A \cap \overline B| \geq r \implies A \cup \overline B = A \veerel A \cup \overline B = \overline B  \veerel |\overline A \cap B| \leq r).
  \end{align*}
\end{lemma}

\begin{proof}
  $(\Longrightarrow)$ We have $A \relbot B$.
  This means that $\clho {\{A,B\}} = \clh {\{A,B\}}$. Let $H=\{A,B\}$.
  Thus, $\clh H = \clh \emptyset \cup \{A, B, \overline A, \overline B\}$ by Lemma \ref{lem:singleton-simple}.
  This equality is useful, as only the closure operator $\clh \cdot$ appears in it.
  Now, we have to do the four cases by hand. We will do the first one in detail and give a sketch for the other three as it is exactly the same proof:
  \begin{itemize}
    \item If $|A \cap B| \geq r$, by rule \rrule 2, $A \cup B \in \clh H$.
    Let us discuss the possible values of $A \cup B$ among hyperedges of $\clh H = \clh \emptyset \cup \{A, B, \overline A, \overline B\}$.
    If $A \cup B \in \clh \emptyset$, then $|A \cup B|\leq r$ or $|\overline{A \cup B}|\leq r$ by Lemma \ref{lem:emptyset-hypergraph}.
    If $A \cup B = \overline A$, then $A=\emptyset$ and $B=V$, meaning that $A \cup B = B$.
    If $A \cup B = \overline B$, then $A=V$ and $B=\emptyset$, meaning that $A \cup B = A$.
    Otherwise, $A \cup B = A$ or $A \cup B = B$.
    Hence we have that:
    \begin{align*}
      |A \cap B| \geq r &\implies |A \cup B|\leq r \veerel |\overline{A \cup B}|\leq r \veerel A \cup B = A \veerel A \cup B = B\\
      &\implies |\overline A \cap \overline B|\leq r \veerel A \cup B = A \veerel A \cup B = B
    \end{align*}
    In the last implication, we have replaced $|\overline{A \cup B}|\leq r$ by $|\overline A \cap \overline B|\leq r$, and we have removed $|A \cup B|\leq r$ because if $|A \cap B| \geq r$ and $|A \cup B| \leq r$, then $A=B$, which implies $A \cup B = A$.
    \item If $|\overline A \cap \overline B| \geq r$, then we have $|\overline{A \cup B}| \geq r$. By property \prule 1, we have $A \cap B \in \clh H$, and by rule \rrule 1, we have $\overline {A \cap B} = \overline A \cup \overline B \in \clh H$.
    With the exact same reasoning as the previous item, by replacing $A$ by $\overline A$, $B$ by $\overline B$, and vice versa,
    we conclude that $(|\overline A \cap \overline B| \geq r \implies \overline A \cup \overline B = \overline A \veerel \overline A \cup \overline B = \overline B  \veerel |A \cap B| \leq r)$.
    \item If $|\overline A \cap B| \geq r$, then we have $|B \setminus A| \geq r$.
    By property \prule 2, we have $A \setminus B \in \clh H$, and by rule \rrule 1, we have $\overline {A \setminus B} = \overline A \cup B \in \clh H$.
    With the exact same reasoning as the first item, by replacing $A$ by $\overline A$ and vice versa,
    we conclude that $(|\overline A \cap B| \geq r \implies \overline A \cup B = \overline A \veerel \overline A \cup B = B  \veerel |A \cap \overline B| \leq r)$.
    \item If $|A \cap \overline B| \geq r$, then we have $|A \setminus B| \geq r$.
    By property \prule 2, we have $B \setminus A \in \clh H$, and by rule \rrule 1, we have $\overline {B \setminus A} = A \cup \overline B \in \clh H$.
    With the exact same reasoning as the first item, by replacing $B$ by $\overline B$ and vice versa,
    we conclude that $(|A \cap \overline B| \geq r \implies A \cup \overline B = A \veerel A \cup \overline B = \overline B  \veerel |\overline A \cap B| \leq r)$.
  \end{itemize}
  $(\Longleftarrow)$ We have to prove that $A \relbot B$.
  Recall that $H=\{A,B\}$, and define $H' := \clho H = \clh \emptyset \cup \{A, B, \overline A, \overline B\}$.
  By Observation \ref{lem:closure-sub}, we know that $\clho H \subseteq \clh H$.
  It remains to prove that $\clh H \subseteq H'$.
  In order to do so, it suffices to prove that $H'$ has $H$ as a sub-hypergraph, and that $H'$ satisfies rules \rrule 0, \rrule 1 and \rrule 2.
  We know that $H'$ has $H$ as a sub-hypergraph by definition of $H$ and $H'$.
  We already know that $H'$ satisfies rules \rrule 0 and \rrule 1 because $H' = \clho H$.
  It remains to prove that $H'$ satisfies \rrule 2.

  Let $C$ and $D$ be two distinct hyperedges of $H'$ such that $|C \cap D| \geq r$, and let us prove that $C \cup D$ is also a hyperedge of $H'$.

  First, note that if $C$ is a hyperedge such that $|C|\leq r$, then $|C \cap D| \leq |C| \leq r$, meaning that $|C \cap D| = r$ and $|C| = r$, and thus $C \subseteq D$.
  Hence, $C \cup D = D$, which is a hyperedge of $H'$.
  Secondly, if $|C|\geq n-r$, then $|C \cup D|\geq n-r$, and $C \cup D$ is a hyperedge of $\clh \emptyset$, and thus a hyperedge of $H'$.
  Hence, if $C\in\clh\emptyset$, then $C\cup D \in H'$.
  By symmetry, if $D\in\clh\emptyset$, then $C\cup D \in H'$.
  Now, it remains the cases where $\{C,D\} \subseteq \{A, B, \overline A, \overline B\}$. Taking into account the symmetry between $C$ and $D$, we have six cases.
  If $C = \overline D$, it means that $\{C,D\} = \{A, \overline A\}$ or $\{C,D\} = \{B, \overline B\}$.
  In both cases, $C \cup D = V \in H'$.
  We are now left with four cases.
  Note that they are proved exactly the same way:
  \begin{itemize}
    \item If $C=A$ and $D=B$, since $|C \cap D|\geq r$, then $|A \cap B|\geq r$.
    By hypothesis, $(|A \cap B|\geq r \implies A \cup B = A \veerel A \cup B = B  \veerel |\overline A \cap \overline B| \leq r)$.
    If $A \cup B = A$, then $C \cup D = A \cup B = A \in H'$.
    If $A \cup B = B$, then $C \cup D = A \cup B = B \in H'$.
    If $|\overline A \cap \overline B| \leq r$, then $|\overline{A \cup B}| \leq r$, and $\overline{A \cup B}\in H'$ by rule \rrule 0, and $A\cup B\in H'$ by rule \rrule 1, \textit{i.e.}$C \cup D \in H'$.
    \item If $C=\overline A$ and $D=\overline B$, since $|C \cap D|\geq r$, then $|\overline A \cap \overline B| \geq r$.
    By hypothesis, $(|\overline A \cap \overline B| \geq r \implies \overline A \cup \overline B = \overline A \veerel \overline A \cup \overline B = \overline B  \veerel |A \cap B| \leq r)$.
    We do the same analysis by replacing $A$ by $\overline A$, $B$ by $\overline B$, and vice versa, and we obtain that $C \cup D \in H'$.
    \item If $C=\overline A$ and $D=B$, since $|C \cap D|\geq r$, then $|\overline A \cap B| \geq r$.
    By hypothesis, $(|\overline A \cap B| \geq r \implies \overline A \cup B = \overline A \veerel \overline A \cup B = B  \veerel |A \cap \overline B| \leq r)$.
    We do the same analysis as the first case by replacing $A$ by $\overline A$ and vice versa, and we obtain that $C \cup D \in H'$.
    \item If $C=A$ and $D=\overline B$, since $|C \cap D|\geq r$, then $|A \cap \overline B| \geq r$.
    By hypothesis, $(|A \cap \overline B| \geq r \implies A \cup \overline B = A \veerel A \cup \overline B = \overline B  \veerel |\overline A \cap B| \leq r)$.
    We do the same analysis as the first case by replacing $B$ by $\overline B$ and vice versa, and we obtain that $C \cup D \in H'$.
  \end{itemize}
  This concludes the proof.
\end{proof}

From Lemma~\ref{lem:ortho-equivalence}, we derive some other equivalences of the relation $A \relbot B$.

\begin{corollary}
  \label{cor:ortho-equivalence-2}
  Let $V=[n]$ and let $A \subseteq V$ and $B \subseteq V$.
  We have the following equivalence.
  \begin{align*}
    &A \relbot B\\
    &\iff \\
    &(|A \cap B| < r \veerel  A \subseteq B \veerel B \subseteq A  \veerel |\overline{A \cup B}| <r \veerel |A \cap B| = |\overline{A \cup B}| = r) \wedgerel \\
    &(|A \setminus B| < r \veerel A \cap B = \emptyset \veerel \overline{A \cup B} = \emptyset \veerel |B \setminus A| <r \veerel |A \setminus B| = |B \setminus A| = r).
  \end{align*}
\end{corollary}

\begin{proof}

  We start from the equivalence of Lemma \ref{lem:ortho-equivalence} and we use De Morgan's law as well as basic set equivalences such as $A \cup B = A \iff B \subseteq A$ and $\overline A \cup B = \overline A \iff A \cap B = \emptyset$.
  Then, let $R_1 \equiv (|A \cap B| < r \veerel  A \subseteq B \veerel B \subseteq A  \veerel |\overline{A \cup B}| <r)$ and let $R_2 \equiv (|A \setminus B| < r \veerel A \cap B = \emptyset \veerel \overline{A \cup B} = \emptyset \veerel |B \setminus A| <r)$.
  Recall that for all propositions $X, Y$, we have $(X \implies Y) \iff (\neg X \veerel Y) $.
  Furthermore, for all propositions $X, Y, Z$, we have $(Z \veerel X) \wedgerel (Z \veerel Y) \iff Z \veerel (X \wedgerel Y)$.
  Hence:
  \begin{align*}
      &A \relbot B\\
      &\iff\\
      &(|A \cap B| < r \veerel A \subseteq B \veerel B \subseteq A  \veerel |\overline{A \cup B}| \leq r) \wedgerel \\
      &(|\overline{A \cup B}| < r \veerel A \subseteq B \veerel B \subseteq A  \veerel |A \cap B| \leq r) \wedgerel \\
      &(|A \setminus B| < r \veerel A \cap B = \emptyset \veerel \overline{A \cup B} = \emptyset \veerel |B \setminus A| \leq r) \wedgerel \\
      &(|B \setminus A| < r \veerel A \cap B = \emptyset \veerel \overline{A \cup B} = \emptyset \veerel |A \setminus B| \leq r)\\
    &\iff\\
    &(R_1 \veerel  |\overline{A \cup B}| = r) \wedgerel (R_1 \veerel  |A \cap B| = r) \wedgerel
    (R_2 \veerel  |B \setminus A| = r) \wedgerel (R_2 \veerel  |A \setminus B| = r)\\
    &\iff\\
    &(R_1 \veerel |A \cap B| = |\overline{A \cup B}| = r ) \wedgerel
    (R_2 \veerel |A \setminus B| = |B \setminus A| = r),
  \end{align*}
  which concludes the proof.
\end{proof}

When $r=1$, Corollary~\ref{cor:ortho-equivalence-2} simplifies to $ A \relbot[1] B \iff (A \subseteq B \veerel B \subseteq A  \veerel A \cap B = \emptyset \veerel \overline{A \cup B} = \emptyset )$, providing that $n>4$.
Indeed, in that case, relations $R_1$ and $R_2$ from the proof become equal and we can factorize by it.
We need $n>4$ to remove $|A \cap B| = |\overline{A \cup B}| = |A \setminus B| = |B \setminus A| = r$ from the factorized equivalence.
Hence, this generalizes the notion of non-crossing hyperedges, as introduced in \cite{cunningham1980combinatorial}.
%
%

The definition of $r$-cross-free hypergraph translates the idea that applying \rrule 2 does not provide any new hyperedge.
Let us go back to Example \ref{ex:cross-free} that shows some non-trivial orthogonal hyperedges.
Using Corollary \ref{cor:ortho-equivalence-2}, we can now prove easily what we said in that example :

\begin{example}
    Let $V=[n]$ with $n=12$, and let $r=3$.
    Let $A=\{1,2,3\}$ and $B=\{2,3,4,5,6\}$.
    By using Corollary \ref{cor:ortho-equivalence-2}, since $|A \cap B| = |\{2,3\}| < r$ and $|A \setminus B| = |\{1\}| < r$, we have that $A \relbot B$.
    Let $C=\{1,2,3,4,5,6\}$ and $D=\{4,5,6,7,8,9\}$.
    By using Corollary \ref{cor:ortho-equivalence-2}, since $|C \cap D| = |\overline{C \cup D}| = r$ and $|C \setminus D| = |D \setminus C| = r$, we have $C \relbot D$.
\end{example}

Before moving on to the next section, we complete this section by listing some properties regarding $r$-orthogonal hyperedges.

\begin{lemma}
  \label{lem:ortho-prop}
  We have the following properties for all hyperedges $A$ and $B$ and $r \geq 0$:
  \begin{enumerate}
    \item \label{bullet:ortho-small} If $|A| \leq r$, then $A \relbot B$.
    \item \label{bullet:ortho-reflex} $A \relbot A$.
    \item \label{bullet:ortho-reflcompl} $A \relbot \overline A$.
    \item \label{bullet:ortho-sym} If $A \relbot B$, then $B \relbot A$.
    \item \label{bullet:ortho-symcompl} If $A \relbot B$, then $A \relbot \overline B$.
    \item \label{bullet:ortho-total} If $A \relbot B$, then $A' \relbot B'$ for $A' \in \{A, \overline A\}$ and $B' \in \{B, \overline B\}$.
    \item \label{bullet:ortho-incr} If $A \relbot B$, then $A \relbot[r+1] B$.
  \end{enumerate}
\end{lemma}

\begin{proof}
  Let us prove each point using Corollary \ref{cor:ortho-equivalence-2}, which states
  $A \relbot B \iff (|A \cap B| < r \veerel  A \subseteq B \veerel B \subseteq A  \veerel |\overline{A \cup B}| <r \veerel |A \cap B| = |\overline{A \cup B}| = r) \wedgerel (|A \setminus B| < r \veerel A \cap B = \emptyset \veerel \overline{A \cup B} = \emptyset \veerel |B \setminus A| <r \veerel |A \setminus B| = |B \setminus A| = r).$
  \begin{enumerate}
    \item There are three cases.
    If $ |A \cap B| = r$, then $A \setminus B = \emptyset$, which means that $A \relbot B$.
    If $ |A \setminus B| = r$, then $A \cap B = \emptyset$, which means that $A \relbot B$.
    Otherwise, $|A \cap B| < r$ and $|A \setminus B| < r$, which means that $A \relbot B$.
    \item If $r>0$, then $A \subseteq A$ and $|A \setminus A| < r$.
    For the case $r=0$, we have $A \subseteq A$ and $|A \setminus A| = |A \setminus A| = r$.
    \item $|A \cap \overline A| < r$ and $A \cap \overline A = \emptyset$.
    \item Definition \ref{def:orthogonal} is symmetrical in $A$ and $B$.
    \item Corollary \ref{cor:ortho-equivalence-2} is stable by replacing $B$ by $\overline B$.
    \item This is a consequence of points 4 and 5.
    \item If a set has size at most $r$, it has also size at most $r+1$.
  \end{enumerate}
  Hence, every point is proven.
\end{proof}



\section{Cross-free hypergraphs and lower bound}
\label{section:bound}

In this section, we introduce $r$-cross-free hypergraphs as a tool to prove a lower bound on the number of hyperedges needed to represent a $r$-closed hypergraph. We also prove that $r$-cross-free hypergraphs are, in a sense, not compressible, as they have roughly the same number of hyperedges as any of their representation.

\begin{definition}
  \label{def:r-cross-free}
  A hypergraph $H$ is said to be {\em $r$-cross-free} if for every pair of hyperedges $(A,B)$,  $A \relbot B$, \textit{i.e.} $A$ and $B$ are $r$-orthogonal.
\end{definition}

If an $r$-cross-free hypergraph $H$ is closed, we know it can be generated by a hypergraph with $\OO(n^{r+1})$ hyperedges.
In this section, we show that the number of hyperedges of $H$ is also $\OO(n^{r+1})$.

The goal of the proof is to go deep inside the structure of a $r$-cross-free closed hypergraph.
By definition, when we take the closure of a $r$-cross-free hypergraph, we do not need to apply \rrule 2 to do so.
Hence, since we apply only rules \rrule 0 and \rrule 1, we will add $\OO(n^r)$ hyperedges and then double the number of hyperedges.
Thus, intuitively, a $r$-cross-free hypergraph and its closure have roughly the same number of hyperedges.
However, we already know that a closed hypergraph $H$ can be generated by a sub-hypergraph $H'$ with $\OO(n^{r+1})$ hyperedges.
In the case of $r$-cross-free closed hypergraph $H$, since $H$ and $H'$ have roughly the same number of hyperedges, $H$ should also have $\OO(n^{r+1})$ hyperedges, showing that $r$-cross-free closed hypergraphs have a small number of hyperedges.

\subsection{Tools}

\begin{lemma}
  \label{lem:singleton-hypergraph}
  Let $H$ be a $r$-cross-free hypergraph with vertex set $V=[n]$, let $A$ be a hyperedge of $H$.
  Let $ \{A\} $ be the hypergraph with vertex set $V$ that has only $A$ as hyperedge.
  Then: $\clh {\{A\}} = \clh \emptyset \cup \{ A, \overline A \}$.
 \end{lemma}

 \begin{proof}
    $\clh \emptyset \cup \{ A, \overline A \}$ is a hypergraph that contains $\{A\}$ as a sub-hypergraph and that satisfies rules \rrule 0, \rrule 1 and \rrule 2.
    Since $\clh{\{A\}}$ is the smallest such hypergraph, we have
    $\clh {\{A\}} \subseteq \clh \emptyset \cup \{ A, \overline A \}$.

    Since $\emptyset \subseteq \{A\}$, then $\clh{\emptyset} \subseteq \clh {\{A\}}$ by monotony.
    Furthermore, $A$ is a hyperedge of $\clh{\{A\}}$ since $\clh{\{A\}}$ must contains $\{A\}$ as a sub-hypergraph.
    Then, since $\clh{\{A\}}$ must satisfy rule \rrule 1, then $\overline A$ must be a hyperedge of $\clh{\{A\}}$, meaning that $\clh \emptyset \cup \{ A, \overline A \} \subseteq \clh  {\{A\}}$, which concludes the proof.
 \end{proof}

\begin{lemma}
  \label{lem:ortho-closure-equiv}
  Let $H$ be a hypergraph with vertex set $V=[n]$, let $A$ and $B$ be two hyperedges of $H$.
  Then: $A \relbot B \iff \clh {\{A, B\}} = \clh {\{A\}} \cup \clh {\{B\}}$.
\end{lemma}

\begin{proof}
    By Definition \ref{def:ortho-def}, $ A \relbot B \iff  \clh {\{A, B\}} =  \clho {\{A, B\}}$.
    By Lemmas \ref{lem:singleton-simple} and \ref{lem:singleton-hypergraph},  we have the following equalities:
    \begin{align*}
        \clho {\{A, B\}} &= \clh \emptyset \cup \{A, B, \overline A, \overline B\}\\
        &= ( \clh \emptyset \cup \{A, \overline A\}) \cup ( \clh \emptyset \cup \{B, \overline B\})\\
        &= \clh {\{A\}} \cup \clh {\{B\}}.
    \end{align*}
    Hence the result.
\end{proof}

\begin{lemma}
  \label{lem:cross-free-decomp}
  Let $H$ be a $r$-cross-free hypergraph with vertex set $V=[n]$ and at least one hyperedge.
  For every hyperedge $A$ of $H$, let $ \{A\} $ be the hypergraph with vertex set $V$ that has only $A$ as hyperedge.
  Then:
  \begin{align*}
    \clh H = \bigcup_{A \in H} \clh {\{A\}}
  \end{align*}
\end{lemma}

\begin{proof}
  One inclusion is easy: for any hyperedge $A$, we have $\{A\} \subseteq H$.
  Hence $\clh{\{A\}} \subseteq \clh H$, meaning that $\bigcup_{A \in H} \clh {\{A\}}\subseteq \clh H$.

  Now, let us prove that $\clh H$ is included in the union $U:= \bigcup_{A \in H} \clh {\{A\}}$. To do so, let us prove that $U$ satisfies rules \rrule 0, \rrule 1 and \rrule 2.
  The first two rules trivially hold.
  For rule \rrule 2, let $A, B$ be two hyperedges of $U$ such that $|A \cap B| \geq r$.
  Hence, there exists $A' \in H$ and $B' \in H$ such that $A \in \clh{\{A'\}}$ and $B\in\clh{\{B'\}}$.
  By Lemma \ref{lem:singleton-hypergraph}, $ \clh{\{A'\}} = \clh \emptyset \cup \{ A', \overline {A'} \}$.
  If $A \in \clh \emptyset$,  by Lemma \ref{lem:emptyset-hypergraph}, either $|A| \leq r$ or $|A| \geq n-r$.
  If $|A| \leq r$, since $|A \cap B| \geq r$, we have $|A| = |A \cap B| = r$, and $A \subseteq B$, meaning that $A \cup B = B \in U$.
  If  $|A| \geq n-r$, $|A \cup B| \geq n-r$ and $A \cup B \in \clh \emptyset \subseteq U$.
  It remains to consider the case where $A \subseteq \{ A', \overline {A'} \}$.
  With the same analysis, we can consider that $B \subseteq \{ B', \overline {B'} \}$.

  Since $H$ is $r$-cross-free, $A \relbot B$. By point \ref{bullet:ortho-total} of Lemma \ref{lem:ortho-prop}, $A' \relbot B'$. Then, by Lemma \ref{lem:ortho-closure-equiv}, $\clh{\{A'\}} \cup \clh{\{B'\}} = \clh{\{A', B'\}}$.
  Hence, $A, B \in \clh{\{A', B'\}}$, and since $|A \cap B| \geq r$, by rule \rrule 2, we have $A \cup B \in  \clh{\{A', B'\}}$.
  Then, since $\clh{\{A', B'\}} = \clh{\{A'\}} \cup \clh{\{B'\}}$, $A \cup B \in \clh{\{A'\}} \cup \clh{\{B'\}} \subseteq U$.
  Since $U$ contains $H$, and since $\clh H$ is the smallest such hypergraph, we have $\clh H \subseteq U$.
\end{proof}

With the three previous lemmas, we can understand the structure of an $r$-cross-free hypergraph.

\begin{lemma}
  \label{lem:cross-free-decomp-full}
  Let $H$ be an $r$-cross-free hypergraph with $n$ vertices. Then we have that:
  \begin{align*}
    \clh H = \clh \emptyset \cup \bigcup_{A \in H} \{ A, \overline A \}.
  \end{align*}
  As a consequence, $|\clh H| \leq 2(r+1)n^r + 2|H|$.
\end{lemma}

\begin{proof}
  By applying Lemmas \ref{lem:cross-free-decomp} and \ref{lem:singleton-hypergraph} to $H$ which is $r$-cross-free, we have:
  \begin{align*}
    \clh H = \bigcup_{A \in H} \clh {H_A} = \bigcup_{A \in H} (\clh \emptyset \cup \{ A, \overline A \}) = \clh \emptyset \cup \bigcup_{A \in H} \{ A, \overline A \}.
  \end{align*}
  Then, since $\clh \emptyset = \{A \subseteq V \st |A| \leq r \text{ or } |\overline A| \leq r \}$, we have that:
  \begin{align*}
    |\clh H| &\leq \sum_{i=0}^{r} \binom n i + \sum_{i=0}^{r} \binom n {n-i} + 2|H|
    \leq 2(r+1)n^r + 2|H|
  \end{align*}
  which concludes the proof.
\end{proof}

As a consequence, when $H$ is $r$-cross-free, $H$ and $\clh H$ have roughly the same number of hyperedges.
We know that $H \subseteq \clh H$, meaning that $|H| \leq |\clh H|$.
With Lemma \ref{lem:cross-free-decomp-full}, we have the following bounds:
$|H| \leq |\clh H| \leq 2(r+1)n^r + 2|H|$.
These bounds can be improved by removing $\clh \emptyset $ from $H$.

\begin{lemma}
  \label{lem:cross-free-inequality}
  Let $H$ be a $r$-cross-free hypergraph with $n$ vertices.
  Then:
  \begin{align*}
      |H \setminus \clh \emptyset| \leq |\clh H \setminus \clh \emptyset| \leq 2|H \setminus \clh \emptyset|.
  \end{align*}
\end{lemma}

\begin{proof}
  Let $H' = H\setminus \clh \emptyset$,
  let $J = H \cap \clh \emptyset$,
  let $U = \bigcup_{A \in H} \{ A, \overline A \}$ and
  let $U' = \bigcup_{A \in H'} \{ A, \overline A \}$. We have:
  \begin{align*}
      U =  \bigcup_{A \in H} \{ A, \overline A \} =  \bigcup_{A \in J} \{ A, \overline A \} \cup  \bigcup_{A \in H'} \{ A, \overline A \} = \bigcup_{A \in J} \{ A, \overline A \} \cup U' \subseteq \clh \emptyset \cup U'.
  \end{align*}
  By Lemma \ref{lem:cross-free-decomp-full}, $\clh H = \clh \emptyset \cup U$.
  Hence, $\clh H \subseteq \clh \emptyset \cup U'$.
  Since $U' \subseteq U$, we also have $\clh \emptyset \cup U' \subseteq \clh \emptyset \cup U = \clh H$.
  All in all, $\clh H = \clh \emptyset \cup U'$.

  According to Lemma \ref{lem:emptyset-hypergraph}, we know that $\clh \emptyset = \{A \subseteq V \st |A| \leq r \text{ or } |A| \geq n-r \}$.
  Since $H' = H \setminus \clh \emptyset$, every hyperedge $A$ of $H'$ satisfies $r < |A| < n-r$.
  Since $U'$ is made only of hyperedges $A$ such that $A\in H'$ or $\overline A\in H'$, every hyperedges $A$ of $U'$ satisfies $r < |A| < n-r$.
  Hence, $U' $ and $ \clh \emptyset $ are disjoint.
  Therefore, the equality $\clh H = \clh \emptyset \cup U'$ can be rewritten as $\clh H \setminus \clh \emptyset = U'$.
  From the definition of $U'$, we have $|H'| \leq |U'| \leq 2|H'|$, which we rewrite as
  $|H \setminus \clh \emptyset| \leq |\clh H \setminus \clh \emptyset| \leq 2|H \setminus \clh \emptyset|$,
  hence the result.
\end{proof}



Now, let us see that the property of being $r$-cross-free is inherited when going back and forth through the closure operator $\clh \cdot$.

\begin{lemma}
  \label{lem:cross-free-equiv}
  Let $H$ be a hypergraph with $n$ vertices.
  $H$ is $r$-cross-free if and only if $\clh H$ is $r$-cross-free.
\end{lemma}

\begin{proof}
  If $H$ is a $r$-cross-free hypergraph, by Lemma \ref{lem:cross-free-decomp-full}, it satisfies:
  \begin{align*}
    \clh H = \clh \emptyset \cup \bigcup_{A \in H} \{ A, \overline A \}. \tag{$*$}
  \end{align*}
  We have to show that every pair of hyperedges $A, B$ of $\clh H$ is orthogonal.
  Let $A, B$ be two hyperedges of $\clh H$.
  If $A \in \clh \emptyset$, then by Lemma \ref{lem:emptyset-hypergraph}, we have $|A| \leq r$ or $|\overline A| \leq r$.
  Hence, with point \ref{bullet:ortho-small} of Lemma \ref{lem:ortho-prop}, $A \relbot B$.
  The same is true if $B \in \clh \emptyset$.
  Thus, we can consider that $A \notin \clh \emptyset$ and $B \notin \clh \emptyset$.
  Hence, using relation $(*)$,
  it means that there exists $A' \in H$ such that $A = A'$ or $\overline{A'}$ and
  there exists $B' \in H$ such that $B = B'$ or $\overline{B'}$.
  Since $H$ is $r$-cross-free, $A' \relbot B'$.
  Using the point \ref{bullet:ortho-total} of Lemma \ref{lem:ortho-prop}, $A \relbot B$.

  If $\clh H$ is $r$-cross-free, then every pair of hyperedge of $\clh H$ is $r$-orthogonal. Since $H \subseteq \clh H$, every pair of hyperedge of $H$ is also $r$-orthogonal, and $H$ is $r$-cross-free.
\end{proof}

\subsection{Bounds}

Using the lemmas stated in the previous section, we can easily prove the following theorem:

\begin{theorem}
  Let $H$ be a hypergraph in $\KK_r(n)$ that is also $r$-cross-free, and let $f$ be a bound on the number of hyperedges needed to represent a hypergraph of order $n$.
  If $f(n) = \Omega(n^r)$, then $|H| = \OO(f(n))$.
\end{theorem}

\begin{proof}
  By Theorem \ref{thm:main-1}, there exists $H'$ with $\OO(n^{r+1})$ hyperedges such that $H = \clh{H'}$.
  By Lemma \ref{lem:cross-free-equiv}, $H'$ is $r$-cross-free.
  Hence, by Lemma \ref{lem:cross-free-decomp-full}, $|H| = |\clh{H'}| \leq 2(r+1)n^r + 2|H'| = \OO(n^r) + \OO(f(n)) = \OO(f(n))$.
\end{proof}

By Theorem \ref{thm:main-1}, there exists such a function $f$ with $f(n) = \Theta(n^{r+1})$.
This means that the number of hyperedges of a $r$-cross-free hypergraph is $\OO(n^{r+1})$.

The remainder of this section is dedicated to proving a lower bound on the number of hyperedges needed to represent a closed hypergraph.

\begin{lemma}
  \label{lem:build-family}
  For all $r>0$, there exists an infinite family of $r$-cross-free hypergraphs $\{ H_n \}$, each with $n$ vertices and a number of hyperedges equal to $\Omega(n^r)$ as $n$ goes to infinity.
\end{lemma}

\begin{proof}
  Let $r>0$. Let $n$ be a multiple of $r+1$, \textit{i.e.} $n=k(r+1)$ for some $k\in\ZZ$.
  We will construct a hypergraph with $n$ vertices such that every pair of hyperedge is $r$-orthogonal.
  To do so, let us consider the vertex set $V := \ZZ/k\ZZ \mathrel\times [r+1]$, where $\ZZ/k\ZZ$ denotes the set of integers modulo $k$.
  One can interpret this set as follows: every vertex is assigned a value between 0 and $k-1$ and a color from the set $[r+1]$.
  The value is used modulo $k$ through computation, which is why it is taken from $\ZZ/k\ZZ$ rather than from $\{ 0, \ldots, k-1 \}$.
  The edge set is defined as follows:
  \begin{align*}
    \EE :=  \Bigg\{ A = \{ (v_i,c_i) \}_{i=1}^{r+1} \in \binom V {r+1} \;\Bigg|\; \{c_i\}_{i=1}^{r+1} = [r+1] \text{ and } \sum_{i=1}^{r+1} v_i \equiv 0 \mod k \Bigg\}
  \end{align*}
  In other words, $\EE$ is the family of all sets made of $r+1$ vertices from $V$ such that:
  \begin{itemize}
    \item vertices have pairwise distinct colors, and
    \item the sum of the value of each vertex equals 0 modulo $k$.
  \end{itemize}
  As there are $r+1$ colors and $r+1$ vertices in a hyperedge, each hyperedge contains exactly one vertex of each color.
  We consider the hypergraph $H_n=(V, \EE)$.

  First, remark that if one picks $r$ vertices with pairwise distinct colors, then there is only one hyperedge that contains these $r$ vertices.
  Indeed, let $(v_i,c_i)$ with $1 \leq i \leq r$ be these $r$ vertices.
  Let $(v_{r+1},c_{r+1})$ be a vertex.
  For $\{ (v_i,c_i) \}_{i=1}^{r+1}$ to be a hyperedge, $c_{r+1}$ must be the missing color from  $\{ (v_i,c_i) \}_{i=1}^r$, and $v_{r+1}$ must be equal to $- \sum_{i=1}^r v_i \mod k$.
  Hence, if $A,B \in \EE$ are two hyperedges of $H$ such that $|A \cap B| = r$, then $A=B$.
  As a consequence, for every pair $(A,B)$ of distinct hyperedges, $A \relbot B$.
  Indeed, we know that $|A \cap B| < r$ since $A \neq B$.
  Then, either $|A \setminus B| < r$, which mean that $A \relbot B$,
  or $|A \setminus B| = r$, which implies that  $|A \cap B| = 1$ and that $|B \setminus A| = r$, meaning that $A \relbot B$.
  As a consequence, $H_n$ is $r$-cross-free.

  Now, let us compute the cardinality of $\EE$.
  Without loss of generality, we can set $c_i=i$.
  Hence, there is no choice to take regarding colors.
  For $1 \leq i \leq r $, $v_k$ can be picked freely among $k$ different vertices.
  For $v_{r+1}$, there is only one choice available.
  Hence, $|\EE| = k^r$. As $n = k(r+1)$, we have:
  $|\EE| = k^r = (r+1)^{-r} \; n^r = \Omega(n^r)$, as $r$ is fixed and $n$ goes to infinity.
\end{proof}

\begin{theorem}
  For a fixed $r>0$, there exists an infinite family of $r$-cross-free hypergraphs $\{ H_n \}$, each with $n$ vertices, such that for any hypergraph $H_n$ and for any hypergraph $H'$, if $\clh {H'} = \clh {H_n}$, then $H'$ has at least $\Omega(n^r)$ hyperedges, with respect to $n$ going to infinity.
\end{theorem}

\begin{proof}
  Let $\{ H_n \}$ be the family defined in Lemma \ref{lem:build-family}.
  Let $H'$ be such that $\clh {H'} = \clh {H_n}$.
  Hypergraph $H_n$ is $r$-cross-free, so $\clh{H_n}$ and $H'$ are also $r$-cross-free by Lemma \ref{lem:cross-free-equiv}.
  Therefore, we can apply Lemma \ref{lem:cross-free-inequality}: $2|H' \setminus \clh \emptyset| \geq |\clh {H'} \setminus \clh \emptyset|$, meaning in particular that $2|H'| \geq |\clh {H'} \setminus \clh \emptyset|$.
  Thus, $2|H'| \geq |\clh {H_n} \setminus \clh \emptyset|$.
  Furthermore, $H_n \subseteq \clh {H_n} \setminus \clh \emptyset$ because $H_n \subset \clh {H_n}$ and because every hyperedge $A$ of $H_n$ have a size satisfying $r < |A| < n-r$, meaning that $H_n \cap \clh \emptyset = \emptyset$.
  Hence, $|\clh {H_n} \setminus \clh \emptyset| \geq |H_n|$.
  Recall that the number of hyperedges of $|H_n|$ satisfies $\Omega(n^r)$.
  Hence, since $2|H'| \geq |H_n|$, the number of hyperedges of $H'$ is $\Omega(n^r)$.
\end{proof}

In summary, some $r$-closed hypergraphs need at least $\Omega(n^r)$ hyperedges to be represented, and our method guarantees at most $\OO(n^{r+1})$ hyperedges.

\section{Concluding remarks}

In this paper, we prove that $r$-closed hypergraphs need at least $\Omega(n^r)$ hyperedges to be represented, and we give a method that guarantees at most $\OO(n^{r+1})$ hyperedges to represent such a hypergraph. The natural question is to reduce the gap between the lower and upper bound:
Does there exists a better method that gives at most $\OO(n^{r})$ hyperedges, or are there some $r$-closed hypergraphs that need at least $\Omega(n^{r+1})$ hyperedges to be represented? Or maybe the correct space complexity stands between $n^r$ and $n^{r+1}$?

\end{document}